\documentclass[letterpaper, 10 pt, conference]{ieeeconf}  

\IEEEoverridecommandlockouts                               

\usepackage{graphics} 
\usepackage{graphicx}
\usepackage{epsfig} 
\usepackage{amsmath} 
\usepackage{amssymb}  

\usepackage{mleftright} 

\usepackage{stfloats}%

\usepackage{tikz} 
\usepackage{pgfplots} 

\usepackage{algpseudocode}
\usepackage{algorithm}
\usepackage{epstopdf}

\usepackage{color}

\usepackage{lipsum}
\usepackage{verbatim}

\usepackage{amsthm}

\usepackage{mathtools}

\newtheorem{theorem}{\bf Theorem} \newtheorem{definition}[theorem]{\bf Definition} 
\newtheorem{lemma}[theorem]{\bf Lemma} \newtheorem{remark}[theorem]{\bf Remark}
  \newtheorem{proposition}[theorem]{\bf Proposition}

\title{\LARGE
Verifying dissipativity properties from noise-corrupted input-state data
}

\author{Anne Koch, Julian Berberich, and Frank Allg\"{o}wer 
\thanks{
This work was funded by Deutsche Forschungsgemeinschaft (DFG, German
Research Foundation) under Germany’s Excellence Strategy - EXC 2075 -
390740016. The authors thank the International Max Planck Research School
for Intelligent Systems (IMPRS-IS) for supporting Julian Berberich and Anne
Koch.
\newline
Anne Koch, Julian Berberich, and Frank Allg\"ower 
are with the Institute for Systems Theory and Automatic Control, University of Stuttgart, 70550 Stuttgart, Germany. 
E-mail: {\tt\small$\{$anne.koch, julian.berberich, frank.allgower$\}$@ist.uni-stuttgart.de}.
}%
}

\begin{document}

\maketitle
\thispagestyle{empty}
\pagestyle{empty}

\begin{abstract}
There exists a vast amount of literature how dissipativity properties can be exploited to design controllers for stability and performance guarantees for the closed loop. With the rising availability of data, there has therefore been an increasing interest in determining dissipativity properties from data as a means for data-driven systems analysis and control with rigorous guarantees. Most existing approaches, however, consider dissipativity properties that hold only over a finite horizon and mostly only qualitative statements can be made in the presence of noisy data. In this work, we present a novel approach to determine dissipativity of linear time-invariant systems from data where we inherently consider properties that hold over the infinite horizon. Furthermore, we provide rigorous guarantees in the case of noisy state measurements.
\end{abstract}
\section{Introduction}
In recent years, there has been an increasing interest in setting up a framework for data-driven systems analysis and control. Such a framework should ideally offer similar insights and guarantees as model-based approaches with the caveat that no explicit description of the system via differential equations is necessary but only measured trajectories of the system, which in some sense are informative enough, are needed. The allurement of such an approach is clear: Measured trajectories of a system are usually easy to obtain whereas deriving a mathematical model can be a cumbersome task. One result that recently gained momentum concerning the search of such a data-driven framework is introduced in \cite{Willems05}. In their seminal paper, the authors prove that the behavior of a linear time-invariant (LTI) system can be described by a data-dependent Hankel matrix if the input of the measured input-output trajectory entails sufficient information, i.e. is persistently exciting of sufficient order. This purely data-driven representation of an LTI system, as layed out in \cite{Berberich2019a} and proven in a state-space framework in \cite{Waarde2020}, hence opens up the development of tools for LTI systems which allow rigorous guarantees for systems analysis and control from data. Such tools include state-feedback design \cite{Persis2019}, robust controller synthesis from noisy input-state trajectories \cite{Berberich2019c}, data-driven model predictive control \cite{Coulson2019,Berberich2019b}, stabilizability and controllability analysis \cite{Waarde2020a} and dissipativity properties from input-output trajectories over the finite time horizon \cite{Maupong2017,Romer2019a,Koch2020}.

Determining dissipativity properties from data not only gives insights to the a priori unknown system, but it also opens the possibility to apply readily available feedback theorems to design controllers that achieve guaranteed closed-loop stability. For examples on stabilizing, robust or distributed control schemes on the basis of control theoretic systems properties, where no additional model knowledge is required, see for example the standard publications \cite{Zames1966, Desoer1975, Schaft2000}. More recent examples where knowledge on the $H_\infty$-norm or passivity index together with additional data lead to cooperative or robust controller design, respectively, can be found in \cite{Sharf2019, Holicki2020}. For these reasons, there have been many different approaches to obtain such control theoretic system properties from data. Some ideas on determining system properties such as the $\mathcal{L}^2$-gain or passivity parameters from data can be found in \cite{Romer2017a,Romer2019b,Sharf2020}. The limitation of these approaches is, in general, that huge amounts of data are required and the computational expenses are immense even for small examples. Most of the existing approaches for determining system properties consider LTI systems. One interesting approach in this regard are iterative schemes, see \cite{Wahlberg2010,Romer2019c} and the references therein, where iteratively applying inputs and measuring the outputs asymptotically reveals the true $\mathcal{L}_2$-gain or passivity parameters, respectively, without any a priori model knowledge. However, the disadvantage or limitation of this method is that iterative experiments might not be possible or at least require additional effort and the respective system property can only be certified over a finite time horizon (cf. definition of $L$-dissipativity in \cite{Maupong2017}). Another more recent approach is a 'one-shot approach' which calculates the respective system property from one input-output trajectory, see \cite{Maupong2017,Koch2020} and the references therein. One limitation of this approach is again that sharp results are only available for system properties holding over the finite time horizon. Furthermore, no guarantees can be provided in the case of noisy data. 

Similar to \cite{Persis2019,Berberich2019c,Waarde2020a} we use one input-state trajectory to represent an LTI system from data in the present paper. The advantage with respect to previous methods is that the system properties are certified over the infinite time horizon and that we provide guarantees on the respective system properties based on one noisy input-state trajectory. 

\section{Problem setup} 
In this paper, we consider multiple-input multiple-output 
discrete-time LTI systems of the form
\begin{align}
\begin{split}
x_{k+1}&=Ax_k+Bu_k,\\ 
y_k&=Cx_k+Du_k,
\label{eq:sys}
\end{split}
\end{align}
with $x_k \in \mathbb{R}^n$, $u_k \in \mathbb{R}^m$ and $y_k \in \mathbb{R}^p$. 
We assume that $A$ and $B$ are unknown, but one input-state trajectory of the system is available.
We collect the resulting input-state sequences
$\{u_k\}_{k=0}^{N-1}$, $\{x_k\}_{k=0}^{N}$ in the following matrices 
\begin{align*}
X &:= \begin{pmatrix} x_0 & x_1 & \cdots & x_{N-1} \end{pmatrix}, \\
X_+ &:= \begin{pmatrix} x_1 & x_2 & \cdots & x_N \end{pmatrix}, \\
U &:= \begin{pmatrix} u_0 & u_1 & \cdots & u_{N-1} \end{pmatrix}.
\end{align*}
Further, we assume that $C$, $D$ are known (or, alternatively, the corresponding output trajectory $\{y_k\}_{k=0}^{N-1}$ is additionally available cf.~Remark~\ref{rem:output}). 

Our approach is hence based on only one measured trajectory of the system with the only assumption that the data, i.e. the measured trajectory, is informative enough. Generally, this can be ensured by requiring that the input of the measured trajectory is persistently exciting in the following sense.
\begin{definition}\label{def:pe}
We say that a sequence $\left\{u_k\right\}_{k=0}^{N-1}$ with $u_k\in\mathbb{R}^m$ is persistently exciting of order $L$, if $\text{rank}\left(H_L(x)\right)=mL$.
\end{definition}
Due to their relevance in systems analysis and control, we are now interested in dissipativity properties of LTI systems~\eqref{eq:sys} on the basis of the available data. While the notion of dissipativity was introduced in \cite{Willems1972} for general (nonlinear) systems, we make use of equivalent formulations for LTI systems with quadratic supply rates as, for example, presented in \cite{Scherer2000}. Quadratic supply rates are quadratic functions $s: \mathbb{R}^m \times \mathbb{R}^p \rightarrow \mathbb{R}$ defined by 
\begin{align}
s(u,y) = \begin{pmatrix} u \\ y \end{pmatrix}^\top \Pi
\begin{pmatrix} u \\ y \end{pmatrix}.
\label{eq:supply} 
\end{align}
The matrix $\Pi \in \mathbb{R}^{(m+p) \times (m+p)}$ will be partitioned as
\begin{align*}
\Pi=\begin{pmatrix}R \phantom{^\top}&S^\top\\S&Q\end{pmatrix}
\end{align*}
throughout this paper with $Q=Q^{\top} \in \mathbb{R}^{m\times m}$, $S \in \mathbb{R}^{p\times m^{\phantom{l}}}$ and $R=R^\top \in \mathbb{R}^{p\times p}$. 
\begin{definition}
A system \eqref{eq:sys} is said to be dissipative with respect to the supply rate $s$ if there exists a function $V: \mathbb{R}^n \rightarrow \mathbb{R}$ such that
\begin{align*}
V(x_{k_1}) - V(x_{k_2}) \leq \sum_{i=k_2}^{k_1 -1} s(u_i,y_i)
\end{align*}
for all $0 \leq k_2 < k_1$ and all signals $(u,x,y)$ which satisfy \eqref{eq:sys}.
\end{definition}
The supply rate and the corresponding matrices $(Q,S,R)$ hereby define the system property of interest. For the supply rates defined by
\begin{align}
\Pi_{\gamma} = \begin{pmatrix} \gamma^2 I & 0 \\ 0 & -I \end{pmatrix}, \quad \Pi_{\text{P}} = \begin{pmatrix} 0 & I \\ I & 0 \end{pmatrix},
\label{eq:pigamma}
\end{align}
for example, we retrieve the operator gain $\gamma$ and the passivity property, respectively. The dissipativity property specified by $(Q,S,R)$ will in the following also be referred to as $(Q,S,R)$-dissipativity. 

The following standard result gives equivalent conditions on dissipativity of an LTI system, which we will in the remainder of the paper make use of to determine dissipativity from data. Explanations and the proofs can be found, e.g. in \cite{Scherer2000,Kottenstette2014} and references therein. 
\begin{theorem}
Suppose that the system~\eqref{eq:sys} is controllable and let $s$ be a quadratic supply rate of the form~\eqref{eq:supply}. Then the following statements are equivalent.
\begin{itemize}
\item[a)] The system~\eqref{eq:sys} is dissipative with respect to the supply rate $s$.
\item[b)] There exists a quadratic storage function $V(x) := x^\top P x$ with $P = P^\top$ 
such that 
\begin{align*}
V(x_{k+1}) - V(x_k) \leq s(u_k,y_k)
\end{align*}
for all $k$ and all $(u,x,y)$ satisfying \eqref{eq:sys}.
\item[c)] There exists a matrix $P = P^\top$ such that 
\begin{align}
\label{eq:diss_lmi}
\begin{pmatrix} A^\top PA - P - \hat{Q} & A^\top PB - \hat{S} \\
(A^\top PB - \hat{S} )^\top & -\hat{R} + B^\top P B \end{pmatrix} \preceq 0
\end{align}
with  $\hat{Q} = C^\top QC$, $\hat{S} = C^\top S + C^\top QD$ and $\hat{R} = D^\top QD + (D^\top S + S^\top D) + R$.
\end{itemize}
\label{thm:diss_lmi}
\end{theorem}

In the following we use these equivalences to verify or find dissipativity properties from data. More precisely, we start in the next section by introducing an equivalent data-based dissipativity formulation on the basis of noise-free input and state trajectories.

\section{Data-driven dissipativity from input-state trajectories}
\label{sec:diss}
With the definitions from the last section, we can directly state necessary and sufficient conditions for dissipativity properties from noise-free input and state trajectories. In this case, verifying dissipativity boils down to checking one simple LMI.
\begin{theorem}
\label{thm:1}
Given noise-free data $\{u_k\}_{k=0}^{N-1}$, $\{x_k\}_{k=0}^{N}$ of a controllable LTI system $G$ and the feasibility problem to find $P=P^\top$ such that
\begin{align}
	M \preceq 0
	\label{eq:opt_allg}
\end{align}
with
\begin{align}
\begin{split}
M &= X_+^\top P X_+ - X^\top P X \\
&- \begin{pmatrix} U \\ CX+DU \end{pmatrix}^\top \begin{pmatrix} R & S^\top \\ S & Q \end{pmatrix} \begin{pmatrix} U \\ CX+DU \end{pmatrix}.
\end{split}
\label{eq:M}
\end{align}
\begin{enumerate}
\item If there exists no $P=P^\top$ such that \eqref{eq:opt_allg} holds, then $G$ is not $(Q,S,R)$-dissipative.
\item If there exists $P=P^\top$ such that \eqref{eq:opt_allg} holds and, additionally, $\mathrm{rank} \begin{pmatrix} X \\ U \end{pmatrix} = n+m$,  
then $G$ is $(Q,S,R)$-dissipative.
\end{enumerate} 
\end{theorem}
\begin{proof} Substituting $X_+ = A X + B U$, the semidefinitness condition in \eqref{eq:opt_allg} can be equivalently written as 
\begin{align}
\begin{pmatrix} X \\ U \end{pmatrix}^\top 
\begin{pmatrix} A^\top PA - P - \hat{Q} & A^\top PB - \hat{S} \\
(A^\top PB - \hat{S} )^\top & -\hat{R} + B^\top P B \end{pmatrix} 
 \begin{pmatrix} X \\ U \end{pmatrix}  
\label{eq:pf_lmi}
\end{align}
with  $\hat{Q} = C^\top QC$, $\hat{S} = C^\top S + C^\top QD$ and $\hat{R} = D^\top QD + (D^\top S + S^\top D) + R$.
\begin{enumerate}
\item If problem~\eqref{eq:opt_allg} is infeasible, this directly implies that \eqref{eq:diss_lmi} is not negative semidefinite for any $P$, i.e. $G$ is not dissipative by Theorem~\ref{thm:diss_lmi}.
\item With full row rank of $\begin{pmatrix} X \\ U \end{pmatrix}$, the semidefiniteness condition \eqref{eq:pf_lmi} in turn implies that \eqref{eq:diss_lmi} holds, which implies dissipativity by Theorem~\ref{thm:diss_lmi}.
\end{enumerate}
\end{proof}

\begin{remark}
\label{rem:rank}
The condition $\mathrm{rank} \begin{pmatrix} X \\ U \end{pmatrix} = n+m$ can be easily checked for a given input and state trajectory. This rank condition can also be enforced by choosing the input $\{u_k\}_{k=0}^{N-1}$ persistently exciting of order $n+1$, cf.~\cite[Corollary 2]{Willems05}.
\end{remark}
\begin{remark}
Since the introduced feasibility problem~\eqref{eq:opt_allg} is linear in $(Q,S,R)$, optimization for finding an 'optimal' or 'tight' system property yields a simple SDP. The problem of minimizing $\gamma^2$ such that~\eqref{eq:opt_allg} for $R=\gamma^2 I$, $S=0$ and $Q=-I$, e.g., yields the $\mathcal{L}_2$-gain. Similar formulations can be found for input and output strict passivity, conic relations or general positive-negative supply rates with $Q + I \preceq 0$.
\end{remark}
\begin{remark}
\label{rem:output}
If $C$ and $D$ are unknown but measurements of the output are available instead, then one can equivalently substitute $Y = CX + DU$ in the feasibility problem \eqref{eq:opt_allg} with $Y := \begin{pmatrix} y_0 & y_1 & \cdots & y_{N-1} \end{pmatrix}$. 

Note that since the feasibility problem \eqref{eq:opt_allg} is linear in $C,D$ it might be interesting for some applications to optimize over $C$ and $D$. Such a scenario could be sensor placement with the goal to maximize the output feedback passivity parameter of agents performing cooperative control tasks. 
\end{remark}
It is particularly interesting that the viewpoint taken in this paper allows to determine dissipativity properties over the infinite horizon from only considerably short data. Furthermore, this viewpoint and the corresponding introduced approach also allow to include robust inference of dissipativity from noisy state trajectories as will be discussed in the next section.
\section{Dissipativity properties from noisy input-state trajectories}
\label{sec:noise}
When working with data and measured trajectories, these trajectories are often affected by noise. We therefore derive guarantees for dissipativity properties from noise corrupted state measurements in this section by using similar ideas to \cite{Berberich2019c}.
Therefore, we consider in this section LTI systems that are disturbed by process noise of the form
\begin{align}
x_{k+1} &= A x_k + B u_k + B_w w_k \\
y_k &= C x_k + D u_k 
\label{eq:sys_noise}
\end{align}
where $w_k \in \mathbb{R}^{m_w}$ represents the noise. We denote by $\{\hat{w}_k\}_{k=0}^{N-1}$ the actual noise sequence which led to the measured input-state trajectories $\{u_k\}_{k=0}^{N-1}$, $\{x_k\}_{k=0}^{N}$. While $\{\hat{w}_k\}_{k=0}^{N-1}$ is generally unknown, we assume that information in form of the following bound on the matrix  
\begin{align*}
\hat{W} = \begin{pmatrix} \hat{w}_0 & \hat{w}_1 & \cdots & \hat{w}_{N-1} \end{pmatrix}
\end{align*}
is available. To be specific, we assume that $\hat{W}$ is an element of the set
\begin{align}
\mathcal{W} = \{ W \in \mathbb{R}^{m_w \times N} | \begin{pmatrix} W \\ I \end{pmatrix}^\top \begin{pmatrix} Q_w & S_w \\ S_w^\top & R_w \end{pmatrix} \begin{pmatrix} W \\ I \end{pmatrix} \succeq 0 \}
\label{eq:W}
\end{align}
with $Q_w \in \mathbb{R}^{m_w \times m_w}$, $S_w \in \mathbb{R}^{m_w \times N}$ and $R_w \in \mathbb{R}^{N \times N}$ with $R_w \succ 0$. Hence, any unknown noise realization that affects the measured data is bounded by a quadratic matrix inequality. This definition of the set $\mathcal{W}$ is a flexible noise or disturbance description in literature (cf.~\cite{Scherer2000,Berberich2019c}).

Since the actual realization of the noise $\{\hat{w}_k\}_{k=0}^{N-1}$ corresponding to the measured input and state trajectories $\{u_k\}_{k=0}^{N-1}$, $\{x_k\}_{k=0}^{N}$ is unknown, there generally exist multiple pairs $(A_{\text{d}},B_{\text{d}})$ which are consistent with the data for some noise instance $W \in \mathcal{W}$. We denote the set of all such $(A_d, B_d)$ by
\begin{align*}
\Sigma_{X,U} = \{ (A_{\text{d}},B_{\text{d}}) | X_+ = A_{\text{d}}X+B_{\text{d}}U+B_w W, W \in \mathcal{W} \}.
\end{align*}
By assumption, the 'true' system matrices $(A,B)$ are in the set $\Sigma_{X,U}$, i.e. $X_+ = AX + BU + B_W \hat{W}$ with $\hat{W} \in \mathcal{W}$.
The key for guaranteeing that a system~\eqref{eq:sys} has a dissipativity property is that we need to verify this dissipativity property for \textit{all} systems which are consistent with the data for some $W \in \mathcal{W}$, i.e. for all systems in the set $\Sigma_{X,U}$. Therefore, we first develop a data-driven open-loop representation in the following lemma.
\begin{lemma}
If there exists a matrix $\mathcal{G}$ such that 
\begin{align}
	\begin{pmatrix} X \\ U \end{pmatrix} \mathcal{G} = I
	\label{eq:G_inv}
\end{align}
then all $(A_d, B_d)$ in the set $\Sigma_{X,U}$ can equivalently be described by
\begin{align}
\begin{pmatrix} A_d & B_d \end{pmatrix} = (X_+ - B_w W) \mathcal{G}
	\label{eq:para}
\end{align}
for any $W \in \mathcal{W}$ satisfying \begin{align}
(X_+ - B_w W)
\begin{pmatrix}
X \\ U 
\end{pmatrix}^\perp
= 0.
\label{eq:necessary}
\end{align}
\label{lem:para}
\end{lemma}
 
\begin{proof}
First note that as explained in \cite[Theorem~4]{Berberich2019c}, the constraint \eqref{eq:necessary} is, by the Fredholm alternative, equivalent to the existence of a solution $V$ to the system of linear equations
\begin{align}
V \begin{pmatrix} X \\ U \end{pmatrix} = X_+ - B_w W.
\label{eq:fredholm}
\end{align}

i) Let us assume \eqref{eq:para} holds for some $W \in \mathcal{W}$ with \eqref{eq:necessary}. We need to show that there exists an $(\tilde{A}, \tilde{B})$, $\tilde{W} \in \mathcal{W}$ such that 
\begin{align}
\begin{pmatrix} \tilde{A} & \tilde{B} \end{pmatrix} = (X_+ - B_w W) \mathcal{G} \label{eq:pf_temp}\\
\text{with} \quad X_+ = \tilde{A}X + \tilde{B}U+ B_w \tilde{W}. \notag
\end{align}

We know that for all $W \in \mathcal{W}$ satisfying \eqref{eq:necessary}, there exists a solution $V$ to \eqref{eq:fredholm}. Hence the choice $\begin{pmatrix} \tilde{A} & \tilde{B} \end{pmatrix} = V$ from \eqref{eq:fredholm}  ensures
\begin{align*}
X_+ = \tilde{A} X + \tilde{B} U + B_w \tilde{W}
\end{align*}
with $\tilde{W} = W$, and $\begin{pmatrix} \tilde{A} & \tilde{B} \end{pmatrix} = V$ also satisfies
\begin{align*}
\begin{pmatrix} \tilde{A} & \tilde{B} \end{pmatrix} = \begin{pmatrix} \tilde{A} & \tilde{B} \end{pmatrix} \begin{pmatrix} X \\ U \end{pmatrix} \mathcal{G} = (X_+ - B_w W) \mathcal{G}.
\end{align*}

ii) For any $(A_d, B_d) \in \Sigma_{X,U}$ per definition there exists a $W \in \mathcal{W}$ such that
\begin{align*}
A_d X + B_d U = X_+ - B_w W.
\end{align*}
This implies the existence of solution V to \eqref{eq:fredholm} hence \eqref{eq:necessary} holds.
Multiplying $\mathcal{G}$ from the right on both sides immediately yields 
\begin{align*}
\begin{pmatrix} A_d & B_d \end{pmatrix} \begin{pmatrix} X \\ U \end{pmatrix} \mathcal{G} = \begin{pmatrix} A_d & B_d \end{pmatrix}
 = (X_+ - B_w W)\mathcal{G}.
\end{align*}
\end{proof}

While in Lemma~\ref{lem:para}, an equivalent description of $\Sigma_{X,U}$ from input and state data has been introduced, we will in the following mainly consider the following superset of $\Sigma_{X,U}$. Let $\Sigma_{X,U}^S$ denote the set of systems which are described by 
\begin{align}
\begin{pmatrix} A_d^s & B_d^s \end{pmatrix} = ( X_+ - B_w W )\mathcal{G} \quad \text{for any} \; W \in \mathcal{W}.
\label{eq:superset}
\end{align}
We hence drop the condition \eqref{eq:necessary}, which immediately shows $\Sigma_{X,U} \subseteq \Sigma_{X,U}^S$.

By introducing the equivalent formulation for $(A_d, B_d)$ on the basis of data in Lemma~\ref{lem:para} and defining the resulting superset $\Sigma_{X,U}^S$ in \eqref{eq:superset}, we have rewritten the problem in a form such that we can directly apply robust systems analysis tools to find sufficient conditions on dissipativity properties. More precisely, the set $\Sigma_{X,U}^S$ can be represented in a linear fractional transformation (LFT) of a nominal system with the disturbance $W$ by 
\begin{align}
\begin{pmatrix}
x_{k+1} \\ y_k \\ z_k 
\end{pmatrix} = \begin{pmatrix} X_+ \mathcal{G} & -B_w \\ \begin{pmatrix}C & D\end{pmatrix} & 0 \\ \mathcal{G} & 0 \end{pmatrix} 
\begin{pmatrix}
x_{k} \\ u_k \\ \tilde{w}_k 
\end{pmatrix}\quad \text{with} \;\; \tilde{w}_k = W z_k
\label{eq:lft}
\end{align}
with $W \in \mathcal{W}$.
This brings us to the main result in this section, which allows to guarantee dissipativity properties from noisy input-state trajectories.

\begin{theorem}
Let $Q \preceq 0$. If there exists a matrix $\mathcal{G}$ with \eqref{eq:G_inv} and $P=P^\top \succ 0$, $\tau > 0$ s.t. \eqref{eq:rob_perf_lmi} holds, then all systems consistent with the data $(A_d, B_d) \in \Sigma_{X,U}$ are $(Q,S,R)$-dissipative. 
\label{thm:noise}
\end{theorem}
\begin{proof}
This result follows from an application of known robust control methods to the system in \eqref{eq:lft}, cf.~\cite{Scherer2000,Scherer2000a}. As $\Sigma_{X,U} \subseteq \Sigma_{X,U}^S$, this proves the claim.
\end{proof}

\begin{figure*}[b]
\noindent\makebox[\linewidth]{\rule{\textwidth}{0.4pt}}
\vspace{2pt}
\begin{align}\label{eq:rob_perf_lmi}
\mleft(
\begin{array}{ccc}I&0&0\\\multicolumn{2}{c}{(X_+ \mathcal{G})\phantom{-}}&-B_w\\\hline
0&I&0\\C&D&0\\\hline0&0&I\\\multicolumn{2}{c}{\mathcal{G} \phantom{,} }&0\end{array}
\mright)^\top
\mleft(
\begin{array}{cc|cc|cc}
-P&0&0&0&0&0\\
0&P&0&0&0&0\\\hline
0&0&-R&-S^\top&0&0\\
0&0&-S&-Q&0&0\\\hline
0&0&0&0&\tau Q_w&\tau S_w\\
0&0&0&0&\tau S_w^\top&\tau R_w
\end{array}\mright)
\mleft(
\begin{array}{ccc}I&0&0\\ \multicolumn{2}{c}{(X_+ \mathcal{G}) \phantom{,}}&-B_w\\\hline
0&I&0\\C&D&0\\\hline0&0&I\\\multicolumn{2}{c}{\mathcal{G} \phantom{,} }&0\end{array}
\mright)
\prec0 
\end{align}
\end{figure*}

\begin{remark}
Requiring $Q \preceq 0$ is necessary to apply the full block S-procedure \cite{Scherer2000} that leads to the result used in the proof of Theorem~\ref{thm:noise}. Note, however, that this includes most relevant dissipativity properties such as the $\mathcal{L}_2$-gain (cf.\ $\Pi_\gamma$ in~\eqref{eq:pigamma}) and passivity (cf.\ $\Pi_\mathrm{P}$ in~\eqref{eq:pigamma}).
\end{remark}  

\begin{remark}
The sufficiency condition in Theorem~\ref{thm:noise} is that there exists a matrix $\mathcal{G}$ such that \eqref{eq:G_inv} holds, i.e. requires that there exists a right-inverse of the matrix $\begin{pmatrix} X \\ U \end{pmatrix}$. This is equivalent to requiring this matrix to have full row rank, i.e. $\mathrm{rank}\begin{pmatrix} X \\ U \end{pmatrix} = n+m$, cf.\ Remark~\ref{rem:rank}.
\end{remark}

\begin{remark}
Note that assuming $B_w$ to be known is not restrictive. Including $B_w$ in the analysis simply offers one approach how additional knowledge on the influence of the process noise can be included into the optimization problem. For example, $B_w$ can easily incorporate knowledge on which states are affected by noise or if some states are affected by the same noise. If no knowledge is available on how the noise acts on the system, one could simply use the identity matrix $B_w = I$.
\end{remark}

Ideally, we would like to include the condition \eqref{eq:necessary} into the optimization problem, which is generally still an open problem and part of ongoing investigations. In the special case that $\begin{pmatrix} X \\ U \end{pmatrix}$ has full rank and quadratic (i.e. $N=n+m$), the condition \eqref{eq:necessary} is trivially satisfied and $\Sigma_{X,U} = \Sigma_{X,U}^S$. In this case, the feasbility problem can be written in a more compact way without equality condition. This result is summarized in the following proposition.
\begin{proposition}
\label{prop:noise}
Let $N=n+m$ and $\begin{pmatrix} X \\ U \end{pmatrix}$ have full rank. 
Then all $(A_d, B_d) \in \Sigma_{X,U}$ are $(Q,S,R)$-dissipative if there exists $P=P^\top$, $\tau > 0$ such that 
\begin{align}
\begin{pmatrix}
B_w^\top P B_w + \tau Q_w& -B_w^\top P X_+ + \tau S_w \\ - X_+^\top P B_w + \tau S_w^\top & M + \tau R_w\end{pmatrix} \preceq 0
\label{eq:opt_noise}
\end{align}
with $M$ as defined in \eqref{eq:M}.
\end{proposition}
\begin{proof} Under constraint qualification (i.e. there exists a $W \in \mathrm{int}(\mathcal{W})$), the S-procedure for two quadratic terms is necessary and sufficient \cite{Boyd2004}. Hence, the feasibility problem in \eqref{eq:opt_noise} can be equivalently formulated as
\begin{align}
\begin{pmatrix}
W \\ I\end{pmatrix}^\top
\begin{pmatrix}
B_w^\top P B_w & -B_w^\top P X_+ \\ - X_+^\top P B_w & M
\end{pmatrix}
\begin{pmatrix}
W \\ I \end{pmatrix} \preceq 0 \quad \forall W \in \mathcal{W}.
\label{eq:opt_noise2}
\end{align}

With $\begin{pmatrix} X \\ U \end{pmatrix}$ quadratic and full rank, \eqref{eq:necessary} is fulfilled for all $W \in \mathcal{W}$. From Lemma~\ref{lem:para} we hence know that any $(A_d, B_d)$ that is consistent with the data can be 
written as
\begin{align}
\begin{pmatrix} A_d & B_d \end{pmatrix} = (X_+ - B_w W )\mathcal{G} \quad \text{for any} \; W \in \mathcal{W}
\end{align}
and thus, for all $(A_d, B_d) \in \Sigma_{X,U}$, $A_d X + B_d U$ can be \textit{equivalently} expressed by 
\begin{align}
A_d X + B_d U = (X_+ - B_w W ) \quad \text{for any} \; W \in \mathcal{W}.
\label{eq:para2}
\end{align}
Since \eqref{eq:para2} is an equivalent reformulation, we can substitute the term $X_+ - B_w W$ for all $W \in \mathcal{W}$ by the term $A_{\text{d}} X + B_{\text{d}} U$ for all $(A_{\text{d}}, B_{\text{d}}) \in \Sigma_{X,U}$ in the feasibility problem \eqref{eq:opt_noise2}. This yields the following condition equal to \eqref{eq:opt_noise2}:
\begin{align}
\begin{split}
\begin{pmatrix} X \\ U \end{pmatrix}^\top 
\begin{pmatrix} A_{\text{d}}^\top PA_{\text{d}} {-} P {-} \hat{Q} & A_{\text{d}}^\top PB_{\text{d}} {-} \hat{S} \\
(A_{\text{d}}^\top PB_{\text{d}} {-} \hat{S} )^\top & -\hat{R} {+} B^\top P B_{\text{d}} \end{pmatrix} 
 \begin{pmatrix} X \\ U \end{pmatrix}  \preceq 0
\end{split}
\label{eq:pf_lmi2}
\end{align}
for all $(A_{\text{d}}, B_{\text{d}}) \in \Sigma_{X,U}$ with  $\hat{Q} = C^\top QC$, $\hat{S} = C^\top S + C^\top QD$ and $\hat{R} = D^\top QD + (D^\top S + S^\top D) + R$. 

Therefore, if \eqref{eq:opt_noise} is feasible and hence \eqref{eq:pf_lmi2} holds for all $(A_{\text{d}}, B_{\text{d}}) \in \Sigma_{X,U}$, then all $(A_{\text{d}}, B_{\text{d}}) \in \Sigma_{X,U}$ are $(Q,S,R)$-dissipative arguing as in the proof of Theorem~\ref{thm:1} since $\begin{pmatrix} X \\ U \end{pmatrix}$ has full rank.
\end{proof}

\begin{remark}
The proof of Proposition~\ref{prop:noise} shows that, under the technical assumption that there exists a $W \in \mathrm{int} (\mathcal{W})$, the feasibility problem~\eqref{eq:opt_noise} is equivalent to 
\begin{align*}
\begin{pmatrix} A_{\text{d}}^\top PA_{\text{d}} {-} P {-} \hat{Q} & A_{\text{d}}^\top PB_{\text{d}} {-} \hat{S} \\
(A_{\text{d}}^\top PB_{\text{d}} {-} \hat{S} )^\top & -\hat{R} {+} B^\top P B_{\text{d}} \end{pmatrix} \preceq 0
\end{align*}
for all $(A_{\text{d}}, B_{\text{d}})$ that are consistent with the data.
This implies that the condition~\eqref{eq:opt_noise} is a necessary and sufficient condition for all $(A_{\text{d}}, B_{\text{d}}) \in \Sigma_{X,U}$ being $(Q,S,R)$-dissipative with a common quadratic storage function $V(x) = x^\top P x$.
\end{remark}
Via Proposition~\ref{prop:noise} we do not need to restrict our attention to $Q\preceq0$ in the special case $N=n+m$. Moreover, the resulting feasibility problem~\eqref{eq:opt_noise} is particularly simple in this case, where no additional equality constraint is required.

\section{Numerical Examples}
In the following, we apply the introduced approach to two numerical examples. We illustrate the influence of the noise bound on the robust dissipativity property in the first example, and we focus on the influence of the data length in the second example.  

\subsection{Example 1}
For the first example, we choose a randomly generated example with a system order of $n=4$, $m=2$ inputs and $p=2$ outputs. We choose an input signal uniformly sampled in the interval $[-1,1]$ and measure the state trajectory over the horizon $N=n+m$. We assume to know a bound on the otherwise unknown noise given by $\| \hat{w} \|_2 \leq \bar{w}$, which implies the bound $\hat{W} \in \mathcal{W}$ for $Q_w = -I$, $S_w = 0$, $R_w = \bar{w}^2 I$. Furthermore, we assume that all states are affected by the process noise and no additional knowledge is available leading to the choice $B_w = I$.
To generate the state measurements, we uniformly sample $\hat{w}$ from the ball $\|\hat{w}\|_2 \leq \bar{w}$ to simulate $N$ time steps of the system.
We now apply Proposition~\ref{prop:noise} to infer the shortage of passivity $s$ of our system, i.e. the minimal $s$ for which the system is $(Q,S,R)$-dissipative for $Q = s I$, $S = \frac{1}{2} I$ and $R = 0$, via a simple SDP without knowledge of $\hat{w}$ but only with the bound $\hat{W} \in \mathcal{W}$. The true shortage of passivity of the system is given by $s = 0.83$. The resulting guaranteed upper bounds $\hat{s}$ on the true shortage of passivity are illustrated in Fig.~\ref{fig:ex1} for different noise bounds $\bar{w}$. 

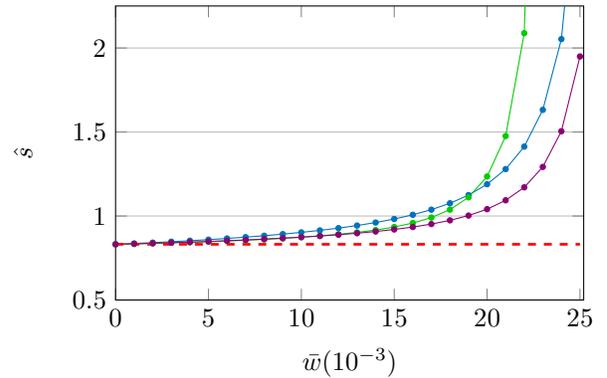
\begin{figure}[h]
\definecolor{mycolor1}{rgb}{0.00000,0.44700,0.74100}%
\begin{tikzpicture}
\begin{axis}[%
width=0.35\textwidth,
height=0.22\textwidth,
at={(0cm,0cm)},
scale only axis,
xmin=0.000,
xmax=25.2,
xlabel={$\bar{w} (10^{-3})$},
ymin=.5,
ymax=2.25,
ylabel={$\hat{s}$},
axis background/.style={fill=white},
ymajorgrids,
]
\addplot [color=mycolor1,mark size=1pt,mark=*,mark options={solid},forget plot]
  table[row sep=crcr]{%
0   0.8316\\
1	  0.8364\\
2		0.8415\\
3		0.8470\\
4		0.8530\\
5		0.8594\\
6		0.8664\\
7		0.8742\\
8		0.8826\\
9		0.8921\\
10	0.9026\\
11	0.9144\\
12	0.9278\\
13	0.9432\\
14	0.9611\\
15	0.9821\\
16	1.0072\\
17	1.0377\\
18	1.0757\\
19	1.1245\\
20	1.1891\\
21	1.2792\\
22	1.4129\\
23	1.6318\\
24	2.0530\\
25	3.1804\\
};
\addplot [color=green!80!black,mark size=1pt,mark=*,mark options={solid},forget plot]
  table[row sep=crcr]{%
0   0.8316\\
1	  0.8342\\
2		0.8371\\
3		0.8402\\
4		0.8437\\
5		0.8474\\
6		0.8515\\
7		0.8561\\
8		0.8613\\
9		0.8672\\
10	0.8739\\
11	0.8818\\
12	0.8911\\
13	0.9024\\
14	0.9164\\
15	0.9343\\
16	0.9580\\
17	0.9907\\
18	1.0380\\
19	1.1112\\
20	1.2350\\
21	1.4760\\
22	2.0882\\
23	5.6619\\
};
\addplot [color=red!55!blue,mark size=1pt,mark=*,mark options={solid},forget plot]
  table[row sep=crcr]{%
0   0.8316\\
1	  0.8345\\
2		0.8377\\
3		0.8411\\
4		0.8447\\
5		0.8486\\
6		0.8529\\
7		0.8574\\
8		0.8624\\
9		0.8680\\
10	0.8741\\
11	0.8809\\
12	0.8887\\
13	0.8975\\
14	0.9078\\
15	0.9199\\
16	0.9343\\
17	0.9520\\
18	0.9741\\
19	1.0026\\
20  1.0406\\
21	1.0934\\
22	1.1707\\
23	1.2922\\
24	1.5047\\
25	1.9496\\
};

\addplot [color=red,solid,forget plot,line width=1pt, dashed]
  table[row sep=crcr]{%
0	0.8316\\
25 0.8316\\
};
\end{axis}
\end{tikzpicture}%
\caption{Robust bound on the shortage of passivity $s$ for a randomly generated $2 \times 2$ system of order $n=4$ for three different randomly sampled noise instances (purple, blue, green) at increasing noise levels.}
\label{fig:ex1}
\end{figure}
Aligned with the theoretical results, the presented approach retrieves the exact shortage of passivity for noise-free measurements and a valid upper bound in the case of noisy state trajectories. The upper bound provided on the shortage of passivity increases with increasing $\bar{w}$, as we require the respective dissipativity property to hold for all systems consistent with the data. The size of this set increases with the noise bound.

\subsection{Example 2}
Next, we randomly generate two systems with two inputs and two outputs and system order $n=6$ and we choose an input signal that is uniformly sampled in the interval $[-1,1]$. We apply Theorem~\ref{thm:noise} to infer the $\mathcal{L}_2$-gain (or, equivalently, $H_\infty$-norm). For this, we first rewrite \eqref{eq:rob_perf_lmi} into an LMI by, first, performing a congruence transformation with $\mathrm{diag}(P^{-1},I)$ and then applying the Schur complement three times. Via a line search over $\tau$ of the resulting LMI, we can then calculate a robust bound on the $\mathcal{L}_2$-gain (with $\Pi_\gamma$ in~\eqref{eq:pigamma}). We calculate such a bound on the $\mathcal{L}_2$-gain for both systems over different data lengths, starting from the minimum length $N=n+m$ and increase $N$ up to $25$. Alternatively, one could also calculate a matrix $\mathcal{G}$ via \eqref{eq:G_inv} and then solve the SDP \eqref{eq:rob_perf_lmi} for that $\mathcal{G}$ (i.e. not optimizing over $\mathcal{G}$, possibly introducing some additional conservatism). 

We assume that we know a bound on $\|\hat{w}_k\|_2 \leq \bar{w}$ that holds for all $k=0,1,\dots,N-1$. This implies the bound $\hat{W} \in \mathcal{W}$ for $Q_w = -I$, $S_w = 0$, $R_w = \bar{w}^2 N I$ with $B_w = I$.
For every system at each time step $k=0,\dots,N-1$, we uniformly sample $\hat{w}_k$ from the ball $\|\hat{w}_k\|_2 \leq \bar{w}$ and we choose $\bar{w} = 0.001$.
In Fig.~\ref{fig:ex2}, we plot the relative difference to the true $\mathcal{L}_2$-gain of the approach, i.e. $\varepsilon = \frac{\hat{\gamma} - \gamma}{\gamma}$. Note that for all results $\hat{\gamma} \geq \gamma$ holds, which means that we indeed always correctly find and verify a $(Q,S,R)$-dissipativity of our system. 

\begin{figure}[h]
\definecolor{mycolor1}{rgb}{0.00000,0.44700,0.74100}%
\begin{tikzpicture}
\begin{axis}[%
width=0.37\textwidth,
height=0.22\textwidth,
at={(0cm,0cm)},
scale only axis,
xmin=7.8,
xmax=25.3,xtick={8,10,15,20,25},
xlabel={$N$},
ymin=0,
ymax=16.5,
ylabel={$\varepsilon (\%)$},
axis background/.style={fill=white},
ymajorgrids,
]
\addplot [color=red!55!blue,mark size=2pt,mark=+,mark options={solid},forget plot]
  table[row sep=crcr]{%
8		37.71\\
9		27.25\\
10	27.58\\
11	15.96\\
12	12.61\\
13	10.19\\
14	5.6\\
15	5.21\\
16	4.92\\
17	4.91\\
18	4.52\\
19	4.44\\
20	4.38\\
21	4.41\\
22	5.15\\
23	4.67\\
24	4.37\\
25	1.14\\
};
\addplot [color=green!80!black,mark size=1pt,mark=o,mark options={solid},forget plot]
  table[row sep=crcr]{%
8		2.9\\
9	  2.93\\
10	3.13\\
11	2.79\\
12	1.59\\
13	1.51\\
14	1.15\\
15	1.15\\
16	1.11\\
17	1.03\\
18	0.91\\
19	.81\\
20	.74\\
21	.74\\
22	.72\\
23	.72\\
24	.69\\
25	.67\\
};

\pgfplotsset{
    after end axis/.code={
				\node[above, text=red!55!blue] at (axis cs:8,16.75){\textcolor{red!55!blue}{\scriptsize{38}}}; 
				\node[above, text=red!55!blue] at (axis cs:9,16.75){\textcolor{red!55!blue}{\scriptsize{27}}}; 
				\node[above, text=red!55!blue] at (axis cs:10,16.75){\textcolor{red!55!blue}{\scriptsize{28}}}; 
    }
}
\end{axis}
\end{tikzpicture}%
\caption{Robust bound on the $\mathcal{L}_2$-gain for two randomly generated $2 \times 2$ system of order $n=6$ for different data lengths $8 \leq N \leq 25$.}
\label{fig:ex2}
\end{figure}
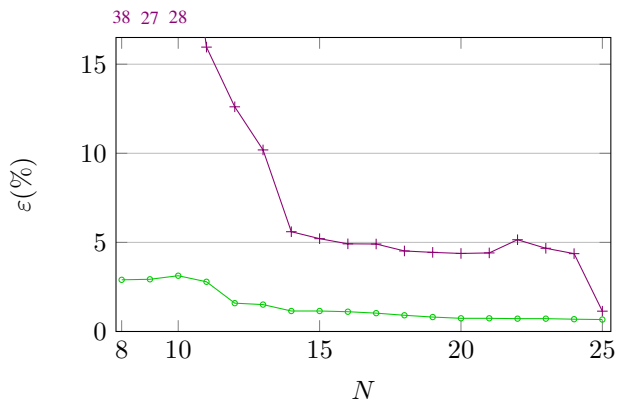

Fig.~\ref{fig:ex2} illustrates that more data points tend to reduce the conservatism as they reduce the size of the set $\Sigma_{X,U}$ and oftentimes also the size of the set $\Sigma_{X,U}^S$. Reducing the size of the set $\Sigma_{X,U}^S$ in turn also reduces $\varepsilon$.

\section{Conclusion and Outlook}
We have introduced a new approach to determine dissipativity properties that hold over the infinite horizon from finite input and state trajectories. We extended this approach by providing guarantees in the case that the input-state trajectories are corrupted by process noise. Numerical examples showed the potential of this method, solving a simple SDP for guaranteed dissipativity property from input and noisy state measurements.

The presented initial results indicate that the taken viewpoint provides advantages over other methods for dissipativity from data and has the potential to be extended with regard to different challenges. Ongoing work includes, for example, the question of determining robust guarantees on dissipativity properties from input-output data. 

\bibliographystyle{IEEEtran}  
\bibliography{bib_all}  

\end{document}